\documentclass[a4paper,10pt]{article}
\pdfoutput=1 

\usepackage{jheppub} 
\usepackage{amssymb,amstext,amsthm,amsfonts,mathrsfs,amsbsy,amsmath,array,multirow}
\usepackage[T1]{fontenc} 
\newtheorem{thm}{Theorem}[section]

\newtheorem{prop}[thm]{Proposition}

\newtheorem{cor}[thm]{Corollary}

\newtheorem{definition}[thm]{Definition}

\title{\boldmath M-theory on non-K\"ahler manifolds}

\author{C. S. Shahbazi}

\affiliation{Institut de Physique Th\'eorique, CEA Saclay.}

\emailAdd{carlos.shabazi-alonso@cea.fr}

\abstract{\vspace{0.1cm}\\ We show that supersymmetric M-theory compactifications to three-dimensional Minkowski space-time preserving $\mathcal{N}=2$ supersymmetry allow for a class of internal manifolds more general than the Calabi-Yau one, namely the class of locally conformally K\"ahler manifolds which locally carry a preferred Calabi-Yau structure.}

\begin{document}
\maketitle
\flushbottom


\section{Introduction}
\label{sec:introduction}


Supersymmetry has been linked in many different and profound ways to geometry since its discovery in the seventies, see for example \cite{2006hep.th....5148Z,Denef:2008wq,2009arXiv0901.1881O,Koerber:2010bx,MooreStrings} for more information and further references. In particular, supersymmetric solutions to Supergravity theories are closely linked to spinorial geometry, since they consist of manifolds equipped with spinors constant respect to a particular connection, whose specific form depends on the Supergravity theory under consideration \cite{Ortin:2004ms,Freedman:2012zz}. The  global existence of spinors and the other Supergravity fields usually constrains the global geometry of the manifold. However, the final resolution of the Supergravity equations of motion usually resorts to the use of adapted coordinates to the problem at a local patch of the manifold. Once we have solved the Supergravity equations of motion, a really hard problem by itself, we have to face another difficulty: in order to fully understand the solution, we need to extract as much information as possible about the global geometry of the manifold just from the existence of some explicit tensors and spinors, which we only know at a local patch. In other words, we want to know which manifolds are compatible with a particular set of tensors and spinors whose form is only known locally. 

In fact, this is not a new problem in Theoretical Physics or Differential Geometry. It was already encountered soon after the discovery of General Relativity. Solving General Relativity's equations of motion\footnote{In contrast with what it is usually implied in the mathematical literature, General Relativity's equations are in general not just the requirement of Ricci-flatness of the Levi-Civita connection.} usually means solving the metric at a local patch of a manifold which is not known a priori. In order to find which would be the physically meaningful manifold compatible with a locally defined metric, physicists back then created a procedure, by now textbook material \cite{Ellis}, to obtain the \emph{maximally analytic extension} of a given local patch endowed with a locally defined metric. In doing so for a \emph{simple} solution, namely the Schwarzschild black-hole, one can find for example that the corresponding manifold can indeed be covered by a single system of coordinates and it is thus homeomorphic to an open set in $\mathbb{R}^{4}$. This procedure has been carried out in other popular solutions of General Relativity, for example the Reissner-Nordstr\"om and the Kerr black holes, which are relatively simple solutions compared to the kind of solutions that one obtains in Supergravity, where finding the maximally analytic extension associated to a local solution is more difficult due to their complexity. 

Still, for supersymmetric solutions of Supergravity some information about the global geometry of the manifold can be obtained simply from the analysis of the existence of constant spinors: for example it may be possible to show that the manifold is equipped with various geometric structures, like Killing vectors or complex, K\"ahler, Hyperk\"ahler, Quaternionic... appropriately defined structures. This already constrains the problem to a relatively specific class of manifolds. However, in performing such analysis sometimes there are involved various kinds of subtle choices, which, if modified, would yield a different global solution, a different manifold which however is locally indistinguishable from the unmodified one, since they exactly carry the same structure at the local level. 

In this note we are going to precisely modify one condition that had been implicitly assumed so far in String-Theory warped compactifications \cite{Grana:2005jc}: we are going to consider that the warp-factor is not a globally defined function on the compact manifold, but only, given a good open covering, locally defined on each open set. In order to do this consistently we will keep in mind that the physical fields of the theory must remain as well-defined tensors on the manifold, as it is required from physical considerations. The warp factor will turn out to be globally described as a section of an appropriate line bundle.

We are going to apply the previous modification to M-theory compactifications to three-dimensional Minkowski space-time preserving $\mathcal{N}=2$ supersymmetries. M-theory compactifications to three dimensions preserving different amounts of supersymmetry have been extensively studied in the literature \cite{Becker:1996gj,Becker:2000jc,Martelli:2003ki,Tsimpis:2005kj,Condeescu:2013hya,Prins:2013wza,Babalic:2014dea,Babalic:2014fua}. In references \cite{Babalic:2014dea,Babalic:2014fua} a very rigorous and complete study of the geometry of the internal eight-dimensional manifold has been carried out using the theory of codimension-one foliations, which turns out to be the right mathematical tool to characterize it, as suggested in \cite{Grana:2014rpa}. 

Coming back to the case of compactifications to three-dimensional Minkowski space-time preserving $\mathcal{N}=2$ supersymmetries, the analysis of the seminal reference \cite{Becker:1996gj} concludes that, among other things, the internal eight-dimensional manifold is a Calabi-Yau four-fold, although the physical metric is not the Ricci-flat metric but conformally related to it. This class of M-theory compactifications is very important for F-theory \cite{Vafa:1996xn} applications, since compactifications of F-theory are in fact defined through them by assuming that the internal manifold is an elliptically fibered Calabi-Yau manifold, see \cite{Weigand:2010wm} for a review and further references.

By assuming that the warp factor is not a global function, we will be able to generalize the result of reference \cite{Becker:1996gj}: we will find that the internal manifold must be a locally conformally K\"ahler manifold \cite{Libermann1,Libermann2,Vaismanlck1976} locally equipped with a preferred Calabi-Yau structure. Evidently, standard Calabi-Yau manifolds are a particular case inside this class. Let us say that this note is of course not the first attempt to extend F-theory compactifications beyond the Calabi-Yau result; see references \cite{Bonetti:2013fma,Bonetti:2013nka} for applications of $Spin(7)$-manifolds to F-theory compactifications.

The consequences of compactifying M-theory on a locally conformally K\"ahler manifold instead of a Calabi-Yau four-fold are manifold since the former is not Ricci-flat in a compatible way and has different topology than the latter. This deserves further study. In particular we think that it would be interesting to obtain, if possible, the effective action of a M-theory compactification on a non-Calabi-Yau locally conformally K\"ahler manifold. 

The plan of this note goes as follows. In section \ref{sec:Mtheory3d} we review, following \cite{Becker:1996gj}, the analysis of M-theory compactifications to three-dimensional Minkowski space-time preserving $\mathcal{N}=2$ supersymmetries. In section \ref{sec:globallocalsusy} we modify the procedure explained in section \ref{sec:Mtheory3d} by considering a warp-factor which is not a globally defined function on the internal manifold. In particular we obtain that, due to this modification, the internal geometry is allowed to be of the locally conformally K\"ahler type. In section \ref{sec:l.c.K.} we elaborate on the locally conformally K\"ahler manifolds obtained in section \ref{sec:globallocalsusy}, comparing them to standard Calabi-Yau four-folds and giving a particular example, the complex Hopf manifold. In section \ref{sec:conclusions} we conclude with some final remarks. 


\section{F-theory compactifications}
\label{sec:Mtheory3d}


In this section we briefly review the standard analysis, following the seminal reference \cite{Becker:1996gj}, of supersymmetric F-theory compactifications, or in other words, supersymemtric compactifications of eleven-dimensional Supergravity to three-dimensional Minkowski space-time preserving $\mathcal{N}=2$ supersymmetry. We will consider the space-time to be an eleven-dimensional oriented spin manifold $M$. The matter content of eleven-dimensional Supergravity \cite{Cremmer:1978km} is then given by a Lorentzian metric $\mathsf{g}$, a closed four-form $\mathsf{G}\in\Omega^{4}(M)$ and the Rarita-Schwinger field $\Psi\in\Gamma(S\otimes T^{\ast}M)$, where $S$ is the rank-32 symplectic spinor vector bundle of $M$, namely a bundle of Clifford-modules. The supersymmetry condition corresponds to the vanishing of the Rarita-Schwinger supersymmetry transformation:

\begin{equation}
\label{eq:supersymmetry}
\delta_{\epsilon}\Psi = D\epsilon = 0\, ,
\end{equation}

\noindent
where $\epsilon\in\Gamma(S)$ is the spinor generating the supersymmetry transformation and $D\colon\Gamma(S)\to\Gamma( T^{\ast}M\otimes S)$ is a particular Clifford-valued connection given in terms of $\mathsf{g}$ and $\mathsf{G}$ and whose explicit form we shall not need here. For F-theory compactifications we consider the space-time to be a topologically trivial product of three-dimensional Minkowski space $\mathbb{R}_{1,2}$ and an eight-dimensional Riemannian, compact, spin manifold $M_{8}$

\begin{equation}
M = \mathbb{R}_{1,2}\times M_{8}\, .
\end{equation}

\noindent
The metric and the four-form are taken to be given by

\begin{eqnarray}
\mathsf{g} &=& \Delta^{2}\, \delta_{1,2} + g\nonumber\\
\mathsf{G} &=& \mathrm{Vol}\wedge\xi + G\, ,
\end{eqnarray}

\noindent
where $\Delta\in C^{\infty}(M_{8})$ is a function, $\delta_{1,2}$ and Vol are the Minkowski metric and the volume form in $\mathbb{R}_{1,2}$, $g$ is the Riemannian metric in $M_{8}$, and $G\in\Omega^{4}(M_{8})$ is a closed four-form in the internal space. Finally, the supersymmetry spinor is decomposed as

\begin{equation}
\epsilon = \chi\otimes\eta\, , \qquad \chi\in\Gamma(S_{1,2})\, , \quad\eta\in\Gamma(S^{\mathbb{C}}_{8})\, , 
\end{equation}

\noindent
where $S_{1,2}$ is the rank-two real spinor bundle over $\mathbb{R}_{1,2}$ and $S^{\mathbb{C}}_{8}$ is the complexified rank-eight spinor bundle over $M_{8}$. Notice that $\eta = \eta_{1}+i\eta_{2}$, where $\eta_{1}, \eta_{2}\in \Gamma(S_{8})$ are Majorana-Weyl spinors of the same chirality. Imposing the previous structure on $M$, together with supersymmetry condition \eqref{eq:supersymmetry}, imposes restrictions on the flux $\mathsf{G}$ and constrains $(M_{8},g)$ at the topological as well as the differentiable level \cite{Becker:1996gj}:

\begin{itemize}


\item  $M_{8}$ is equipped with a $SU(4)$-structure induced by $\eta_{1}$ and $\eta_{2}$, which we assume everywhere independent and non-vanishing. The topological obstruction for the existence of nowhere vanishing real spinor, or in other words, the existence of a $Spin(7)$-structure is given by

\begin{equation}
\label{eq:Spin7obstruction}
p^{2}_{1} - 4p_{2} +8\chi(M_{8}) = 0\, ,
\end{equation}

\noindent
where $p^{2}_{1}$ and $p_{2}$ are the integrated $P^{2}_{1}$ and $P_{2}$ Poyntriagin classes, and $\chi(M_{8})$ is the Euler characteristic of $M_{8}$.


\item $M_{8}$ is equipped with a globally defined almost complex structure $J$, a real non-degenerate (1,1)-form $\omega = g\cdot J$ and a (4,0)-form $\Omega$ constructed as bilinears of $\eta$. The quadruplet

\begin{equation}
\left\{g, J, \omega, \Omega\right\}
\end{equation}

\noindent
makes $M_{8}$ into an almost hermitean manifold with topologically trivial canonical bundle.


\item Let us make the following conformal transformation

\begin{equation}
\label{eq:conformalglobal}
\tilde{g} = \Delta g\, , \qquad \tilde{\eta} = \Delta^{-\frac{1}{2}}\eta\, ,
\end{equation}

\noindent
which implies

\begin{equation}
\tilde{J} =  J\, , \qquad \tilde{\omega} = \Delta \omega\, , \qquad \tilde{\Omega} = \Delta^{2}\Omega\, .
\end{equation}

\noindent
The usefulness of this conformal transformation comes from the fact that the transformed spinors are parallel with respect to the transformed connection, namely 

\begin{equation}
\label{eq:constanttildeeta}
\tilde{\nabla}\tilde{\eta} = 0\, ,
\end{equation}

\noindent
where $\tilde{\nabla}$ is the Levi-Civita connection associated to $\tilde{g}$. Equation \eqref{eq:constanttildeeta} automatically implies that $M_{8}$ has $SU(4)$-holonomy and thus $M_{8}$ is a Calabi-Yau four-fold. In particular we have

\begin{equation}
\tilde{\nabla} \tilde{J} = 0\, , \qquad \tilde{\nabla} \tilde{\omega} = 0\, , \qquad \tilde{\nabla}\tilde{\Omega} = 0\, . 
\end{equation} 

\noindent
We can also see that $M_{8}$ is a Calabi-Yau four-fold as follows, which might be more natural from the algebraic-geometry point of view: $\left(\tilde{g},\tilde{\omega},\tilde{J}\right)$ is the compatible triplet of a complex structure $\tilde{J}$, a symplectic structure $\tilde{\omega}$ and a Riemannian metric $\tilde{g}$ making $M_{8}$ into a K\"ahler manifold. Since $\tilde{\Omega}$ is an holomorphic (4,0)-form, the canonical bundle is holomorphically trivial, which together with the K\"ahler property of $M_{8}$, implies that it is a Calabi-Yau four-fold. 

 
\item The one-form $\xi$ is given by de derivative of the warp factor $\Delta$ as follows

\begin{equation}
\label{eq:xiglobal}
\xi = d\left( \Delta^{3}\right)\, ,
\end{equation}

\noindent
and the four-form $G$ is subject to the constraint

\begin{equation}
\label{eq:Gequationglobal}
\iota_{v} G\cdot \eta = 0\, , \qquad v\in\mathfrak{X}(M_{8})\, .
\end{equation}

\noindent
Once we know that $M_{8}$ is a Calabi-Yau four-fold, equation \eqref{eq:Gequationglobal} can be solved by taking $G$ to be $(2,2)$ and primitive. 
\end{itemize}

\noindent
From the previous analysis we conclude that  if we take $M_{8}$ to be a Calabi-Yau manifold, $G\in H^{(2,2)}(M_{8})$ and primitive and $\xi$ as in equation \eqref{eq:xiglobal}, we solve the supersymmetry conditions \eqref{eq:supersymmetry} and we obtain a supersymmetric compactification background of eleven-dimensional Supergravity to three-dimensional Minkowski space. Note that the physical metric $g$ is conformally related to the Ricci-flat metric $\tilde{g}$, and by Yau's theorem we know that this is the unique Ricci-flat metric in its K\"ahler class, and thus it is, strictly speaking, the Calabi-Yau metric of $M_{8}$.


\section{Global patching of the local supersymmetry conditions}
\label{sec:globallocalsusy}


In this section we are going to slightly generalize the set-up reviewed \ref{sec:Mtheory3d} by considering a situation where the conformal transformation \eqref{eq:conformalglobal} cannot be performed globally, but only locally. We will be still satisfying the eleven-dimensional Supergravity supersymmetry conditions, which are local, but globally we will be able to construct a manifold that is not necessarily a Calabi-Yau four-fold but of a more general type.

As we did in section \ref{sec:Mtheory3d}, we will consider the space-time to be a topologically trivial product of three-dimensional Minkowski space $\mathbb{R}_{1,2}$ and an eight-dimensional Riemannian, compact spin manifold $M_{8}$

\begin{equation}
M = \mathbb{R}_{1,2}\times M_{8}\, .
\end{equation}

\noindent
The supersymmetry spinor is also decomposed exactly as it was done in section \ref{sec:Mtheory3d}, namely

\begin{equation}
\epsilon = \chi\otimes\eta\, , \qquad \chi\in\Gamma(S_{1,2})\, , \quad\eta\in\Gamma(S^{\mathbb{C}}_{8})\, , 
\end{equation}

\noindent
Hence, as it happened in section \ref{sec:Mtheory3d}, $M_{8}$ is equipped with two everywhere independent and non-vanishing  Majorana-Weyl spinors, which implies again that the structure group of $M_{8}$ can be reduced to $SU(4)$. Therefore $M_{8}$ still has to satisfy the obstruction \eqref{eq:Spin7obstruction}.

Let $ \left\{ U_{a}\right\}_{a\in I}$ be a good open covering of $M_{8}$ and let us equip every open set $U_{a}$ with a function $\Delta_{a}\in C^{\infty}(M_{a})$ and a closed one form $\xi_{a}\in\Omega^{1}(U_{a})$, so we can consider the triplet $\left\{ U_{a}, \Delta_{a},\xi_{a}\right\}_{a\in I}$ on $M_{8}$. We will assume that the Lorentzian metric $\mathsf{g}$ and the four-form $\mathsf{G}$ can be written, for every open set $U_{a}\subset M$, as follows:

\begin{eqnarray}
\mathsf{g}|_{U_{a}} &=& \Delta^{2}_{a}\, \delta_{1,2} + g|_{U_{a}}\, ,\nonumber\\
\mathsf{G}|_{U_{a}} &=& \mathrm{Vol}\wedge\xi_{a} + G|_{U_{a}}\, ,
\end{eqnarray}

\noindent
where $g$ is a Riemannian metric in $M_{8}$ and $G$ is a closed four-form in $M_{8}$. In order to keep a clean exposition, we are not explicitly writing the atlas that we are using for $\mathbb{R}_{1,2}$, which, for each $U_{a}$ consists of an open set which we take to be the whole $\mathbb{R}_{1,2}$ and its corresponding coordinate system $\phi_{a}$. More precisely, the atlas that we are considering for the topologically trivial product $M=\mathbb{R}_{1,2}\times M_{8}$ is the following:

\begin{equation}
\mathcal{A} = \left\{V_{a}\times U_{a}, \phi_{a}\times\psi_{a}\right\}_{a\in I}\, ,
\end{equation}

\noindent
where $V_{a} = \mathbb{R}_{1,2}$ for every $a\in I$, $\phi_{a}$ are the coordinates in $V_{a}$ and $\psi_{a}$ are the corresponding local coordinates in $U_{a}$. The atlas $\mathcal{A}$ is obviously not the simplest atlas for $M$, but anyway it is an admissible atlas which gives $M$ the structure of a differentiable product manifold. We will see in a moment that the consistency of the procedure requires very specific changes of coordinates $\phi_{a}\circ\phi^{-1}_{b}\colon\mathbb{R}^{3}\to\mathbb{R}^{3}$, $U_{a}\cap U_{b}\neq\emptyset$. The one-form $\xi$ is given again by equation \eqref{eq:xiglobal}, only this time the result is valid locally in $U_{a}$:

\begin{equation}
\xi_{a} = d\left( \Delta^{3}_{a}\right)\, .
\end{equation}

\noindent
Now, in order for the physical fields $(\mathsf{g},\mathsf{G})$ to be well defined, they must be tensors on $M$. This is equivalent to, given any another open set $U_{b}$ such that $U_{a}\cap U_{b}\neq \emptyset$, the following condition in $U_{a}\cap U_{b}$:

\begin{eqnarray}
\label{eq:patichinggG}
\Delta^{2}_{a}\, \delta_{1,2} + g|_{U_{a}\cap U_{b}} &=& \Delta^{2}_{b}\, \delta_{1,2} + g|_{U_{a}\cap U_{b}}\, ,\nonumber\\ 
\mathrm{Vol}\wedge\xi_{a} + G|_{U_{a}\cap U_{b}} &=& \mathrm{Vol}\wedge\xi_{b} + G|_{U_{a}\cap U_{b}}\, .
\end{eqnarray}

\noindent
Equation \eqref{eq:patichinggG} is equivalent to:

\begin{eqnarray}
\label{eq:patichinggGII}
\Delta^{2}_{a}\, \delta_{1,2} &=& \Delta^{2}_{b}\, \delta_{1,2} \, ,\nonumber\\ 
\mathrm{Vol}\wedge\xi_{a}  &=& \mathrm{Vol}\wedge\xi_{b} \, .
\end{eqnarray}

\noindent
in $U_{a}\cap U_{b}$, up to of course a change of coordinates, which in turn is reflected as a symmetry of the equations of motion. Therefore, we must define the difference between $\Delta_{a}$ and $\Delta_{b}$ in $U_{a}\cap U_{b}$ to be such that it can be absorbed by means of a coordinate transformation in $\mathbb{R}_{1,2}$. The only possibility is: 

\begin{equation}
\label{eq:Deltadifference}
\Delta_{a} = \lambda_{ab} \Delta_{b}\, , 
\end{equation}

\noindent
in $U_{a}\cap U_{b}$, where $\lambda_{ab}\colon U_{a}\cap U_{b}\to \mathbb{R}$ is a constant function. Indeed, the multiplicative factor \eqref{eq:Deltadifference} can be absorbed by means of the following change of coordinates in $\mathbb{R}_{1,2}$: 

\begin{eqnarray}
\phi_{a}\circ\phi^{-1}_{b}\colon\mathbb{R}^{3} &\to &\mathbb{R}^{3}\nonumber\, , \\
x &\mapsto & \lambda^{-1}_{ab} x\, ,
\end{eqnarray}

\noindent
which is of course a diffeomorphism. It can be easily seen that

\begin{equation}
\lambda_{ba} = \lambda^{-1}_{ab}\, , \qquad \lambda_{ab}\lambda_{bc}\lambda_{ca} = 1\, ,
\end{equation} 

\noindent
where the second equation holds in $U_{a}\cap U_{b}\cap U_{c} \neq\emptyset$. Therefore, the following data

\begin{equation}
\left\{M_{8}, U_{a}, \lambda_{ab}, \mathbb{R}\right\}\, ,
\end{equation}

\noindent
defines a flat line bundle $L\to M_{8}$ over $M_{8}$ with connection that descends to a well defined closed one form $\varphi$ in $M_{8}$, namely $[\varphi]\in H^{1}(M_{8})$. Using $L$ we can write the families $\left\{ \Delta_{a} \right\}_{a\in I}$ and $\left\{ \xi_{a}\right\}_{a\in I}$ as 

\begin{equation}
\Delta\in C^{\infty}(M,L)\simeq\Gamma(L)\, , \qquad \xi\in \Omega^{1}(M,L^{3})\, ,
\end{equation}

\noindent
or in other words, in terms of a section of the line bundle $L$ and a one-form taking values in $L^{3}$. 


\subsection{The global geometry of $M_{8}$}
\label{sec:globalM8}


As it happened in section \ref{sec:Mtheory3d}, $M_{8}$ is equipped with a globally defined almost complex structure $J$, a real non-degenerate (1,1)-form $\omega = g\cdot J$ and a (4,0)-form $\Omega$, where $J$ and $\Omega$ are constructed as a bilinears from $\eta$. The quadruplet

\begin{equation}
\label{eq:gJomegaOmegaglobal}
\left\{g, J, \omega, \Omega\right\}
\end{equation}

\noindent
makes $M_{8}$ into an almost hermitean manifold with topologically trivial canonical bundle.

The crucial difference from the situation that we encountered in section \ref{sec:Mtheory3d} is that the conformal  transformation \eqref{eq:conformalglobal} cannot be performed globally. Therefore, we cannot perform the conformal transformation that \emph{transforms} the quadruplet $\left\{g, J, \omega, \Omega\right\}$ into a Calabi-Yau structure in  $M_{8}$, which thus cannot be taken to be a Calabi-Yau four-fold; in particular, the supersymmetry complex spinor is not parallel respect to any Levi-Civita connection associated to a metric in the conformal class of the physical metric. We can however perform the conformal transformation locally on ever open set $U_{a}$, and thus we define

\begin{equation}
\label{eq:conformallocal}
\tilde{g}_{a} = \Delta_{a}g|_{U_{a}}\, , \qquad \tilde{\eta}_{a} = \Delta^{-\frac{1}{2}}_{a}\eta|_{U_{a}}\, ,
\end{equation}

\noindent
where now $\tilde{g}_{a}$ and $\tilde{\eta}_{a}$ are locally defined on $U_{a}$. The local conformal transformation \eqref{eq:conformallocal} implies, again locally in $U_{a}$, that

\begin{equation}
\label{eq:conformallocalII}
\tilde{J}_{a} =  J|_{U_{a}} \, , \qquad \tilde{\omega}_{a} = \Delta_{a} \omega|_{U_{a}}\, , \qquad \tilde{\Omega}_{a} = \Delta^{2}_{a} \Omega|_{U_{a}}\, .
\end{equation}

\noindent
Notice that $J$ is invariant and thus its conformal transformed is a well defined tensor on $M_{8}$. An alternative characterization of these locally defined objects is through globally defined tensors taking values on the corresponding powers of the flat line bundle $L$, namely

\begin{equation}
\tilde{g}\in\Gamma(S^{2}T^{\ast},L)\, , \quad\tilde{\eta}\in\Gamma(S^{\mathbb{C}}_{8},L^{\frac{1}{2}})\, , \quad\tilde{\omega}\in\Omega^{2}(M_{8},L)\, , \quad\tilde{\Omega} \in \Omega^{4}(M_{8},L^{2})\, .
\end{equation} 

\noindent
Once we go to the locally transformed spinor and metric, we have that

\begin{equation}
\label{eq:constanttildeetalocal}
\tilde{\nabla}^{a}\tilde{\eta}_{a} = 0\, ,
\end{equation}

\noindent
where $\tilde{\nabla}^{a}$ is the Levi-Civita connection associated to $\tilde{g}_{a}$ in $U_{a}$. Equation \eqref{eq:constanttildeetalocal} automatically implies, again locally, in $U_{a}$, that

\begin{equation}
\label{eq:localJomegaOmega}
\tilde{\nabla}^{a} \tilde{J}_{a} = 0\, , \qquad \tilde{\nabla}^{a} \tilde{\omega}_{a} = 0\, , \qquad \tilde{\nabla}^{a}\tilde{\Omega}_{a} = 0\, . 
\end{equation} 

\noindent
Hence, we can think of  $\left\{ \tilde{g},\tilde{\omega}_{a},\tilde{J}_{a},\tilde{\Omega}_{a}\right\}$, as a sort of preferred local Calabi-Yau structure in $U_{a}$, which however does not extend globally to $M_{8}$. We can withal obtain globally defined differential conditions in $M_{8}$ which, as we will see later, implies that the geometry of $M_{8}$ belongs to a particular class of locally conformally K\"ahler manifolds. Notice that $\tilde{J}$ is a well-defined almost-complex structure; nonetheless it is not parallel since the Levi-Civita connection in \eqref{eq:localJomegaOmega} is only defined locally in $U_{a}$, as $\tilde{g}_{a}$ is only locally defined in $U_{a}$. In spite of this, we can prove the following:

\begin{prop}
$M_{8}$ is an Hermitian manifold with Hermitian structure $(g,J)$.
\end{prop}

\begin{proof}
Let $N$ denote the Nijenhuis tensor associated to $J$. Then, on every open set $U_{a}\subset M_{8}$ we can locally write $N$ as follows

\begin{equation}
N|_{U_{a}}(u,v) = (\tilde{\nabla}^{a}_{u}J)(J v ) - (\tilde{\nabla}^{a}_{v} J)(J u) + (\tilde{\nabla}^{a}_{Ju}J)(v) - (\tilde{\nabla}^{a}_{Jv}J) (u)\, ,\qquad u,v\in\mathfrak{X}(M_{8})\, ,
\end{equation}

\noindent
and thus $N|_{U_{a}} = 0$ since $J$ is parallel respect to the locally defined Levi-Civita connection $\tilde{\nabla}_{a}$. Since this can be performed in every open set of the covering $\left\{ U_{a}\right\}_{a\in I}$ of $M$, we conclude that $N=0$ and hence $J$ is a complex structure. Since the metric $g$ is compatible with $J$, $(M_{8},g,J)$ is an Hermitian manifold.
\end{proof}

\noindent
Hence, we conclude that $M_{8}$ is a complex Hermitian manifold. There is another global condition that we can extract from \eqref{eq:localJomegaOmega} and which will further restrict the global geometry of $M_{8}$. Equation \eqref{eq:localJomegaOmega} implies that on every open set $U_{a}$ we can find a function, namely $\Delta_{a}$, such that

\begin{equation}
\label{eq:localomega}
d(\Delta_{a} \omega)|_{U_{a}} = 0\, .
\end{equation}

\noindent 
The key point now is that the de-Rahm differential does not depend on the locally-defined Levi-Civita connection $\tilde{\nabla}^{a}$ and therefore we can actually extend equation \ref{eq:localomega} to an equivalent, well-defined, global condition in $M_{8}$. Equation \eqref{eq:localomega} is equivalent to 

\begin{equation}
d\omega|_{U_{a}} + d\log\,\Delta_{a}\wedge\omega|_{U_{a}} = 0\, .
\end{equation}

\noindent
Given another open set $U_{b}$ such that $U_{a}\cap U_{b}\neq \emptyset$ we have that $\log\Delta_{a} = \log\Delta_{b} +\log\lambda_{ab}$ at the intersection and therefore $d\log\Delta_{a} = d\log\Delta_{b}$. Hence, there is a well-defined closed one-form $\varphi\in\Omega^{1}(M_{8})$ such that

\begin{equation}
d\omega = \varphi\wedge\omega\, ,
\end{equation}

\noindent
which is defined, in every open set $U_{a}$, as

\begin{equation}
\varphi|_{U_{a}} = d\log \Delta_{a}\, .
\end{equation}

\noindent
Therefore $\omega$ is a locally conformal symplectic structure \cite{VaismanSymplectic} on $M_{8}$ and thus we have proven the following result:

\begin{thm}
\label{thm:l.c.K.f.}
Let $M_{8}$ be a compact, $SU(4)$-structure locally conformally K\"ahler manifold with locally conformally Ricci-flat K\"ahler metric and locally conformally parallel $(4,0)$-form. Then, $M_{8}$ is an admissible supersymmetric internal space for a supersymmetric compactification of eleven-dimensional Supergravity to three-dimensional Minkowski space-time preserving $\mathcal{N}=2$ supersymmetry. 
\end{thm}

\noindent
The closed one-form $[\varphi]\in H^{1}(M_{8})$, which is usually called the Lee-form, is precisely a flat connection in $L\to M_{8}$. Alternatively, one can define the $\varphi$-\emph{twisted} differential $d_{\varphi} = d-\varphi$ whose corresponding cohomology $H^{\ast}_{\varphi}(M_{8})$ is isomorphic to $H^{\ast}(M_{8},\mathcal{F}_{\varphi})$, the cohomology of $M_{8}$ with values in the sheaf of local $d_{\varphi}$-closed functions. Very good references to learn about locally conformally K\"ahler geometry are the book \cite{lcKbook} and the review \cite{2010arXiv1002.3473O}. 


\subsection{Solving the $G$-form flux}
\label{sec:solvingG}


In order to fully satisfy supersymmetry, we have to impose on the four-form $G$ the constraint

\begin{equation}
\label{eq:GconditionII}
\iota_{v} G\cdot \eta = 0\, , \qquad v\in\mathfrak{X}(M_{8})\, .
\end{equation}

\noindent
In the Calabi-Yau case, this constraint was solved by taking $G$ to be $(2,2)$ and primitive. In our case $M_{8}$ is not a Calabi-Yau manifold but it is a Hermitian manifold and hence it is equipped with a complex structure $J$ and a compatible metric $g$. This turns out to be enough, as we will see now, to conclude that indeed the same conditions, namely $G$ to be $(2,2)$ and primitive, solve equation \eqref{eq:GconditionII} in our case. 

First of all, since we will use this fact later, notice that taking into account that $\eta$ has positive chirality then equation \eqref{eq:GconditionII} implies that $G$ is self-dual in $M_{8}$. Using the Clifford algebra $Cl(8,\mathbb{R})$ relations together with the expresion of $g$ as bilinear of $\eta$, it can be shown that \cite{Becker:1996gj}:

\begin{equation}
\Gamma_{\bar{a}}\eta = \Gamma^{a}\eta = 0\, ,
\end{equation} 

\noindent 
where $\left\{ \Gamma_{a}\right\},\, a=1,\hdots 8$ are the gamma matrices generating $Cl(8,\mathbb{R})$ and the bar denotes an antiholomorphic index. Then, equation \eqref{eq:GconditionII} is equivalent to

\begin{equation}
\label{eq:GconditionIII}
G_{m\bar{a}\bar{b}\bar{c}}\Gamma^{\bar{a}\bar{b}\bar{c}}\eta + 3 G_{m\bar{a}\bar{b}\bar{c}}\Gamma^{\bar{a}\bar{b}c}\eta = 0\, .
\end{equation}

\noindent
The vanishing of the first term in equation \eqref{eq:GconditionIII} is equivalent to

\begin{equation}
G_{m\bar{a}\bar{b}\bar{c}}\eta = 0\, ,
\end{equation}

\noindent
which implies

\begin{equation}
G^{4,0} = G^{3,1} = G^{1,3} = G^{0,4} = 0\, .
\end{equation}

\noindent
The vanishing of the second term in equation \eqref{eq:GconditionIII} is equivalent to

\begin{equation}
\label{eq:Gprimitive}
G_{a\bar{b}c\bar{d}}g^{c\bar{d}} = 0\, .
\end{equation}

\noindent
Taking now into account that $G$ is self-dual, we can rewrite equation \eqref{eq:Gprimitive} as 

\begin{equation}
G\wedge J = 0\, ,
\end{equation}

\noindent
and hence we finally conclude that $G$ is primitive and $G\in H^{(2,2)}(M_{8})$.


\subsection{The tadpole-cancellation condition}
\label{sec:tadpolecancellation}


In order to allow for a non-zero $G$-flux in $M_{8}$, we have to consider $\alpha^{\prime}$-corrections to eleven-dimensional Supergravity, due to the well-known no-go theorem of reference \cite{Maldacena:2000mw}. We will perform the calculation in this section in order to illustrate that although $\left\{\xi_{a}\right\}_{a\in I}$ is not a well-defined one-form in $M_{8}$, due to the fact that $\mathsf{G}$ is an \emph{honest} tensor in $M$, the calculation can be carried out, and since $M_{8}$ is topologically $Spin(7)$, we obtain the same result as in the standard case. The relevant correction for our purposes is given by \cite{Duff:1995wd}

\begin{equation}
\delta S = -T_{M2}\int_{M} C_{3}\wedge X_{8}\, ,
\end{equation}

\noindent
where $G_{4} = dC_{3}$ and $X_{8}$ is an eight-form given by

\begin{equation}
X_{8} = \frac{1}{(2\pi)^{4}}\left( \frac{1}{192} \mathrm{tr}\, R^{4} - \frac{1}{768} \left(\mathrm{tr}\, R^2\right)^{2}\right)\, .
\end{equation}

\noindent
The corrected equation of motion for the four-form $\mathsf{G}$ adapted to the compactification background and written on $M_{8}$ reads

\begin{equation}
\label{eq:adaptedGeq}
\frac{3}{2}d\ast\varphi= -\frac{1}{2} G\wedge G +\beta X_{8}\, ,
\end{equation}

\noindent
where $X_{8}$ can be rewritten in terms of the first and second Pontryagin forms of the internal manifold \cite{AlvarezGaume:1983ig}

\begin{equation}
X_{8} = \frac{1}{192}\left( P^{2}_{1} - 4 P_{2}\right)\, ,
\end{equation}

\noindent
and $\beta$ is an appropriate constant that we will not need explicitly. Notice that $\varphi$ is a one-form locally given by the derivative of the corresponding local warp factor but it cannot be written globally as the derivative of a function, yet it is a well defined closed one-form in $M_{8}$. Asuming that $M_{8}$ is closed, we integrate equation \eqref{eq:adaptedGeq} to obtain

\begin{equation}
\frac{1}{2\beta}\int_{M_{8}} G\wedge G = \int_{M_{8}} X_{8}\, .
\end{equation}

\noindent
Using now that $M_{8}$ has a $SU(4)$-structure and in particular it satisfies equation \eqref{eq:Spin7obstruction}, we obtain

\begin{equation}
\label{eq:tadpole}
\frac{1}{2\beta}\int_{M_{8}} G\wedge G = \frac{\chi(M_{8})}{24} \, ,
\end{equation}

\noindent
a result that was to be expected since it only depends on $M_{8}$ being equipped with a $Spin(7)$-structure. 


\section{Locally conformally K\"ahler manifolds}
\label{sec:l.c.K.}


We have obtained that the supersymmetric conditions on an eleven-dimensional Supergravity compactification to three-dimensional Minkowski space-time, locally preserving $\mathcal{N}=2$ Supersymmetry, allow for locally Ricci-flat, $SU(4)$-structure, locally-conformal K\"ahler manifolds as internal spaces. It is first convenient to introduce the following definition:

\begin{definition}
A $n$-complex dimensional locally conformal Calabi-Yau manifold is $SU(n)$-structure locally-conformal K\"ahler manifold with locally Ricci-flat Hermitian metric and locally conformally parallel $(n,0)$-form.
\end{definition}

\noindent
Hence, the kind of $SU(4)$-structure locally-conformal K\"ahler manifolds that we have obtained as admissible M-theory compactification backgrounds are precisely locally conformal Calabi-Yau manifolds, which motivates the definition. These are not necessarily Calabi-Yau four-folds (which would be a special subclass) and thus it is worth characterizing their geometry. First of all let us summarize the main properties of a generic compact locally conformal Calabi-Yau manifolds:

\begin{enumerate}

\item $M$ is a compact Hermitian manifold. In other words, it is a complex manifold with a Riemannian metric $g$ compatible with the complex structure $J$ of the manifold.

\item $M$ is equipped with non-degenerate two-form $\omega$ constructed from$J$ and $g$, which is not closed but satisfies

\begin{equation}
d\omega = \varphi\wedge\omega\, .
\end{equation}

\noindent
Then $M$ is a particular case of almost-K\"ahler manifold.

\item Although $\omega$ is not closed, locally one can always transform it such that the locally transformed two-form is closed. Therefore $M$ is a particular case of locally conformally symplectic manifold \cite{VaismanSymplectic}.

\item The Riemannian metric $g$ is not Ricci-flat. Despite of this, locally one can find a Ricci-flat metric locally conformal to $g$.

\item There is a globally defined complex spinor which is not parallel respect to the Levi-Civita connection associated to $g$. However, we can make a conformal transformation on the spinor such that it becomes locally parallel respect to the Levi-Civita connection associated to the locally transformed metric.

\item $M$ is equipped with a $SU(n)$-structure, or in other words, it has zero first Chern class in $\mathbb{Z}$. However, the canonical bundle is not holomorphically trivial, as the $(n,0)$-form that topologically trivializes it is not holomorphic, but only locally conformally holomorphic. 

\item $M$ is not projective, in contrast to the Calabi-Yau case. This seemingly technical detail is important, since for example, algebraic-geometry tools are very much used in order to study F-theory on elliptically-fibered Calabi-Yau four-folds.
 
\end{enumerate}

\noindent
There are in the literature several definitions of Calabi-Yau manifolds, not always equivalent. For definiteness, and in order to compare compact Calabi-Yau manifolds with compact locally conformal Calabi-Yau manifolds,  we will use the following two equivalent definitions

\begin{itemize}

\item A compact Calabi-Yau manifold is a compact manifold of real dimension $2n$ equipped with a metric of holonomy equal or contained in $SU(n)$. 

\item A compact Calabi-Yau manifold is a compact K\"ahler manifold with holomorphically trivial canonical bundle. 

\end{itemize}

\noindent
From the previous definitions we see that a locally conformal Calabi-Yau manifold fails to be Calabi-Yau by only two conditions, namely they are not K\"ahler and they do not have an holomorphic $(n,0)$-form, although they are equipped with a $(n,0)$-form topologically trivializing the canonical bundle. The deviation from Calabi-Yau can be measured by $\varphi$, namely, $M$ is Calabi-Yau if and only if $[\varphi]$ is the zero class in de Rahm cohomology. Hence, we have obtained the following result:

\begin{cor}
A simply-connected locally conformal Calabi-Yau manifolds is a Calabi-Yau manifold.
\end{cor}

\noindent
Contrary to what happens with compact locally irreducible Calabi-Yau manifolds, compact locally conformally Calabi-Yau manifolds can have continuous isometries. Let us consider the case of a generic locally conformally K\"ahler manifold $M$: it is equipped with two canonical vector fields $v$ and $u$ given by

\begin{equation}
g(w,u) = \varphi(Jw)\, ,\qquad g(w,v) = \varphi(Jw)\, , \qquad \forall w\in\mathfrak{X}(M)\, .
\end{equation}  

\noindent
Then, the following result holds \cite{VaismanI}:

\begin{prop}
\label{prop:canonicalkilling}
The canonical vector field $u$ is a Killing vector field on $M$ if and only if is an infinitesimal automorphism of $J$, and in this case on has $[u,v]=0$.
\end{prop}

\noindent
Therefore we see that if $u$ is a Killing vector field, then $u$ and $v$ commute and thus they are the infinitesimal generators of a $\mathbb{R}\times\mathbb{R}$-action on $M$. This is a nice starting point to end-up having a torus action and therefore an elliptic fibration on $M$, as explained in proposition 6.4 of \cite{Grana:2014rpa}, where the necessary and sufficient conditions for $u$ and $v$ to define an elliptic fibration where obtained. 

Now that we know that locally conformal Calabi-Yau manifolds are not necessarily Calabi-Yau, an explicit example of a non-Calabi-Yau locally conformal Calabi-Yau manifold is in order. A general a locally conformally K\"ahler manifold $M$ can be written has follows \cite{VaismanI}:

\begin{equation}
M = \tilde{M}/G\, ,
\end{equation}

\noindent
where $\tilde{M}$ is a simply connected K\"ahler manifold, and $G$ is a covering transformation group whose elements are conformal for the respective K\"ahler metric on $\tilde{M}$. This restricts the class of manifolds we can consider, but it is not enough to specify a manageable class. Fortunately, it turns out that there is a class of locally conformally K\"ahler manifolds that has been completely characterized, namely those whose local K\"ahler metric is flat, thanks to the following proposition \cite{VaismanI}:

\begin{thm}
Let $M$ be a compact locally conformally K\"ahler-flat manifold of complex dimension $n$. Then the universal covering space of $M$ is $\mathbb{C}^{n}\backslash \left\{ 0\right\}$, and up to a global conformal change of the metric, $M$ is a generalized Hopf manifold with the canonical metric. Every such manifold $M$ has the same Betti numbers as the Hopf manifold $H^{n}$ of the same complex
dimension $n$. 
\end{thm}

\noindent
A generalized Hopf manifold is a locally conformally K\"ahler manifold such that its Lee-form is a parallel form. Among the generalized Hopf manifolds are of course the classical Hopf manifolds. Four-dimensional complex Hopf manifolds are equipped with a $SU(4)$-structure and thus it is an example of a non-trivial compact locally conformal Calabi-Yau manifold. In particular the metric of a Hopf manifold is not only locally Ricci-flat but locally flat. 

Let us explore then the geometry of compact complex Hopf manifolds, since they provide us with a non-trivial example of locally conformal Calabi-Yau manifolds.


\subsection{Complex Hopf manifolds}
\label{sec:Hopfmanifold}


In this section we are going to introduce complex Hopf manifolds, mainly to give an explicit example of the class of manifolds found in this note, that so far has been only abstractly characterized. The main goal is to obtain the explicit Lee form of a locally conformally Calabi-Yau manifold, since it globally encodes all the locally defined warp factors, and locally gives an explicit expression for the corresponding warp factor of the compactification.

A complex Hopf manifold of complex dimension $n$ is the quotient of $\mathbb{C}^{n}\backslash \left\{ 0\right\}$ by the free action of the infinite cyclic group $\mathfrak{G}_{\lambda}$ generated by $z\to \lambda z$, where $\lambda\in\mathbb{C}^{\ast}$ and $0<|\lambda| <1$. In other words, it is $\mathbb{C}^{n}\backslash \left\{ 0\right\}$ quotiented by the free action of $\mathbb{Z}$, where $\mathbb{Z}$, with generator $\lambda$ acting by holomorphic contractions. The group $\mathfrak{G}_{\lambda}$ acts freely on $\mathbb{C}^{n}\backslash \left\{ 0\right\}$ as a properly discontinuous group of complex analytic transformations of $\mathbb{C}^{n}\backslash \left\{ 0\right\}$, so that the quotient space

\begin{equation}
CH^{n}_{\lambda} = \left( \mathbb{C}^{n}\backslash \left\{ 0\right\}\right)/\mathfrak{G}_{\lambda}\, ,
\end{equation}

\noindent
is a complex manifold of dimension $n$. Complex $n$-dimensional Hopf manifolds $CH^{n}_{\lambda}$ are diffeomorphic to $S^{2n-1}\times S^{1}$. Therefore we have

\begin{equation}
b^{1}\left( CH^{n}_{\lambda}\right) = 1\, ,
\end{equation}

\noindent
and hence $CH^{n}_{\lambda}$ does not admit a K\"ahler metric. In particular, the standard K\"ahler structure on $\mathbb{C}^{n}\backslash \left\{ 0\right\}$ does not descend to $CH^{n}_{\lambda}$. It can be however endowed with a locally conformally K\"ahler structure. To see this, let us consider $\mathbb{C}^{n}\backslash \left\{ 0\right\}$ equipped with the following metric and (1,1)-form written in standard coordinates as follows: 

\begin{equation}
g_{0} = \frac{dz^{t}\otimes d\bar{z}}{\bar{z}^{t}z}\, , \qquad\omega_{0} = i\frac{dz^{t}\wedge d\bar{z}}{\bar{z}^{t}z}\, . 
\end{equation}

\noindent
The (1,1)-form $\omega_{0}$ is not closed but satisfy:

\begin{equation}
d\omega_{0} = \varphi_{0}\wedge\omega_{0}\, , 
\end{equation}

\noindent
where 

\begin{equation}
\varphi_{0} = \frac{z^{t}d\bar{z}+\bar{z}^{t}dz}{\bar{z}^{t}z}\, ,
\end{equation}

\noindent
Since $g_{0}$, $\omega_{0}$ are scale-invariant, they descend to a well defined metric $g$ and (1,1)-form $\omega$ in $CH^{n}_{\lambda}$, with corresponding Lee-form $\varphi$. Notice that in $\mathbb{C}^{n}\backslash \left\{ 0\right\}$ we have that $\varphi_{0}$ is exact, since

\begin{equation}
\varphi_{0} = d\left( \log z^{t}\bar{z} + c\right)\, , \qquad c^{\prime}\in\mathbb{R}\, ,
\end{equation}

\noindent
as was to be expected, since $(g_{0},\omega_{0})$ is globally conformal to the standard K\"ahler structure on $\mathbb{C}^{n}\backslash \left\{ 0\right\}$. Nonetheless, $\varphi$ is not exact in $CH^{n}_{\lambda}$, since there $\log z^{t}\bar{z}$ is not a globally defined function. Let $(U_{a}, z_{a})$ be a coordinate chart in $CH^{n}_{\lambda}$. Then, in $(U_{a}, z_{a})$ we have that 

\begin{equation}
\varphi|_{U_{a}} = d \left(\log z^{t}_{a}\bar{z}_{a} + c^{\prime}_{a}\right)\, , \qquad c^{\prime}_{a}\in\mathbb{R}\, ,
\end{equation}

\noindent
and thus the warp factor of eleven-dimensional Supergravity compactified on $CH^{4}_{\lambda}$ is, at every coordinate chart $(U_{a}, z_{a})$, given by

\begin{equation}
\Delta_{a} = c_{a} z^{t}_{a}\bar{z}_{a}\, , \qquad c\in \mathbb{R}^{\ast}\, ,
\end{equation}

\noindent
and hence, from the geometry of $CH^{4}_{\lambda}$ and the structure of the compactification we have been able to obtain the corresponding warp factor. Notice that $CH^{4}_{\lambda}$ is locally conformally isometric to euclidean space equipped with the standard flat metric. Therefore it trivially satisfies the supersymmetry conditions of eleven-dimensional Supergravity. Notice that the Euler characteristic of $CH^{n}_{\lambda}$ is zero, and thus the tadpole cancellation condition \eqref{eq:tadpole} cannot be satisfied for non-zero flux $G$. However this can be easily fixed by blowing up at a point, since this procedure increases the Euler characteristic and preserves the locally conformally K\"ahler property of the manifold \cite{Tricceri,2009arXiv0906.1657V}. 

This has been only one relatively simple example of a locally conformally K\"ahler manifold, which in particular is locally flat. However the class of locally conformally Calabi-Yau manifolds is much richer and nowadays a full classification is out of reach.


\section{Conclusions and final remarks}
\label{sec:conclusions}


In this note we have shown that it is possible to obtain a class of manifolds more general than the Calabi-Yau one as internal spaces of M-theory supersymmetric compactifications to three-dimensional Minkowski space-time preserving $\mathcal{N}=2$ supersymmetry with chiral internal spinors. In order to go beyond the Calabi-Yau result, we considered a warp factor which is defined only locally and does not extend to a globally defined function on the internal manifold $M_{8}$. In order to do this consistently, it is necessary to use a non-trivial cover of the external Minkowski space with specific coordinate transformations. The physical fields remain as well-defined tensors on the space-time.

We have chosen to name the new class of manifolds obtained in this letter as locally conformally Calabi-Yau manifolds, which form a particular class of locally conformally K\"ahler manifolds, already studied in the literature, characterized by the fact that locally, the metric is conformally Ricci-flat. 

Given this new class of compactification backgrounds, it is natural to wonder if there are physical consequences of considering F-theory compactified on them instead of a standard Calabi-Yau four-fold. 

For starters, the tadpole-cancellation condition \eqref{eq:tadpole} and the conditions on the $G$-flux are formally the same in both cases. However, the topology of Calabi-Yau manifolds and locally conformally Calabi-Yau manifolds is very different. In particular, the Hodge numbers are different and do not satisfy the same relations that hold in the Calabi-Yau case. For instance, the first Betti number $b^{1}$ of a locally conformally Calabi-Yau is non-zero, and thus we encounter the first new feature of this class of backgrounds respect to the standard Calabi-Yau ones: they have different Dolbeault-cohomology, something that has important consequences on the effective action of the corresponding compactification. Notably, the decomposition of the Supergravity fields in terms of basis of harmonic forms of the different cohomology groups will be different and thus the number of moduli and the geometry of the effective scalar manifold will be different.

In fact, locally conformally K\"ahler manifolds have a very restricted moduli space, which implies that the effective theory will contain very few moduli, possibly even none. This is a very exciting possibility from the phenomenological point of view. Unfortunately, very little is currently known about the moduli space of locally conformally K\"ahler structures on a given manifold\footnote{Private communication by Liviu Ornea.}.

For F-theory applications, we need $M_{8}$ to be elliptically fibered in order to be able to make the lift from three-dimensional to four dimensionsal Minkowski space-time. In the standard Calabi-Yau case, being a compact, Ricci-flat, locally irreducible manifold, it has no continuous isometries and there are no smooth non-trivial elliptic fibrations. The presence of singularities is then related to the presence of a particular number of D7-branes, which back-react to the geometry and as a consequence the profile of the axiodilaton, the complex structure of the torus-fibre at every point in the base, become necessarily non-trivial. When there are no D7-branes there is no backreaction and thus the axiodilaton is constant and the elliptic fibration is trivial.

However, the situation for compact locally conformally Calabi-Yau manifolds is very different: they are not Ricci-flat and therefore they may admit isometries and in particular, non-trivial elliptic fibrations without singularities, namely principal $T^{2}$-bundles. As a matter of fact, proposition \ref{prop:canonicalkilling} shows that if the canonical vector field on a locally conformally K\"ahler manifold is Killing, then there is a pair of commuting vector fields on the manifold. This is precisely the right starting point for this pair of vector fields to define an infinitesimal torus action on the manifold, as explained in \cite{Grana:2014rpa}.

Locally conformally K\"ahler manifolds are expected to admit singular non-trivial elliptic fibrations, although this is still an open question: one of the novelties here respect to the Calabi-Yau case is the presence of also regular non-trivial elliptic fibrations. The interpretation, if any, within the context F-theory of non-singular non-trivial elliptic fibrations is unclear.

Let us end by saying that the procedure used in this note can be extended to other String/M-theory compactifications, and as a result we can obtain new, unexpected, admissible compactification backgrounds.


\acknowledgments

I would like to thank Pablo Bueno, Mariana Gra\~na, Thomas W. Grimm, Dominic Joyce, Ruben Minasian, Tom Pugh, Raffaele Savelli and Marco Zambon for very useful discussions and comments. This work was supported in part by the ERC Starting Grant 259133 -- ObservableString.


\appendix


\renewcommand{\leftmark}{\MakeUppercase{Bibliography}}
\phantomsection
\bibliographystyle{JHEP}
\bibliography{C:/Users/cshabazi/Dropbox/Referencias/References}
\label{biblio}
\end{document}